\documentclass[3p,11pt]{elsarticle}

\usepackage{amssymb}
\usepackage{amsthm}
\usepackage{amssymb}

\usepackage{amsmath}
\usepackage{tikz}
\usepackage{mathdots}
\usepackage{yhmath}
\usepackage{cancel}
\usepackage{color}
\usepackage{siunitx}
\usepackage{array}
\usepackage{multirow}
\usepackage{amssymb}
\usepackage{gensymb}
\usepackage{tabularx}
\usepackage{booktabs}
\usepackage{xcolor}

\newtheorem{theorem}{Theorem}
\newtheorem{lemma}[theorem]{Lemma}

\usepackage{pgf, tikz}
\usetikzlibrary{arrows, automata, petri, fadings}

\usepackage{longtable}
\usepackage{todonotes}

\usepackage{algorithm}
\usepackage[noend]{algpseudocode}
\usepackage[utf8]{inputenc}

\usepackage{subcaption}
\usepackage{amsmath}
\usepackage{amsfonts}

\newcommand{\BO}[1]{\mathcal{O}\left(#1\right)}

\newcommand{\BT}[1]{\Theta\left( #1\right)}

\newcommand{\predc}[1]{\mathrm{pre}\left(#1\right)}

\newcommand{\dein}{d}

\newcommand{\gra}{G=\left(V,E\right)}

\newcommand{\mymin}[2][]{\min_{#1}\left\{#2\right\}}
\newcommand{\mymax}[2][]{\max_{#1}\left\{#2\right\}}
\newcommand{\mycurl}[1]{\left\{#1\right\}}
\newcommand{\mybrace}[1]{\left(#1\right)}

\newcommand{\myfloor}[1]{\lfloor #1\rfloor}

\newcommand{\boundary}[1]{\mathcal{B}\left(#1\right)}

\DeclareMathOperator*{\argmax}{arg\,max}

\journal{ArXiv
}

\begin{document}

\begin{frontmatter}



\title{A universal bound on the space complexity\\ of Directed Acyclic
  Graph computations}

\author{Gianfranco Bilardi}
\ead{bilardi@dei.unipd.t}
\affiliation{organization={University of Padova, Department of Information Engineering}
    }

\author{Lorenzo De Stefani}
\ead{lorenzo_destefani@brown.edu}
\affiliation{organization={Brown University, Department of Computer Science}}

\begin{abstract}
It is shown that $S(G) = \BO{m/\log_2 m + \dein}$ pebbles are sufficient to pebble any DAG $\gra$, with $m$ edges and
maximum in-degree $\dein$.  It was previously known that $S(G) = \BO{\dein n/\log n}$. The result builds on two novel ideas. The first is the notion of $B$-\emph{budget decomposition} of a DAG $G$, an efficiently computable partition of $G$ into at most $2^{\lfloor \frac{m}{B} \rfloor}$ sub-DAGs, whose cumulative space requirement is at most $B$. The second is the \emph{challenging vertices} technique, which constructs a pebbling schedule for $G$ from a pebbling schedule for a simplified DAG $G'$, obtained by removing from $G$ a selected set of vertices $W$ and their incident edges. This technique also yields improved pebbling upper bounds for DAGs with bounded genus and for DAGs with bounded topological depth.
\end{abstract}







\end{frontmatter}

\pagebreak
\section{Introduction}

The ``\emph{Pebble
Game}''~\cite{friedman1971algorithmic,paterson1970comparative}, has
been proposed to analyze the memory space requirements of straight-line algorithms whose execution can be modeled computation
Directed Acyclic Graphs (DAGs), where vertices represent both inputs
and functional operations and edges represent data dependencies among
them. A pebble on a vertex indicates that the value associated with
that vertex is in memory. A \emph{pebbling schedule} is a sequence of
moves that place or remove pebbles according to certain rules reviewed
below. The ``\emph{pebbling cost}'' of a pebbling schedule is the
maximum number of pebbles simultaneously present on vertices of the
DAG, at any stage of that schedule.  The ``\emph{Space complexity}''
$S\left(G\right)$ (or ``\emph{Pebbling Number}'') of a DAG $G$ is the minimum
pebble cost of a schedule that, with no pebbles initially on the DAG,
eventually places a pebble on each vertex at least once.

The pebble game and its variations (\emph{e.g.}, the Red-Blue~\cite{jia1981complexity}, Black-White~\cite{cook1974storage}, and
two-person pebble game~\cite{Dymond1983SpeedupsOD}) have been widely studied in the literature,
with applications to modeling of non-deterministic computation,
communication complexity, and proof of work.  A connection between the analysis of the pebbling number and the I/O complexity of computational DAGs is developed in~\cite{bilardi2000space, Bilardi2022TheDV}.

Improving on a result by Paul, Tarjan, and Celoni~\cite{paul1976space}, Lengauer and Tarjan introduced the ``\emph{FAST-PEBBLE}" procedure~\cite{lengauer1979upper}. Given any DAG $G$ with $n$ vertices and maximum in-degree $d$,  and given any $S \geq c_1 d \log_2 d n/\log_2 n$, \emph{FAST-PEBBLE} constructs a pebbling schedule for $G$ using at most $S$ pebbles in $T\leq S(c_2d)^{c_3^{(d+1)\frac{n}{S}}}$  moves,  where $c_1>0$, $c_2>1$, $c_3>1$ are sufficiently high constants. In~\cite{loui1979minimum}, Loui presented a different pebbling strategy using $S\geq \left(3d +4\right)\frac{n}{\log n}$ pebbles in $T\leq S2^{\BO{n/S}}$ moves, based on the notion of \emph{layered partition} of a DAG. A \emph{layered partition} with the desired properties is shown to exist but is not explicitly constructed.

This work establishes an $\BO{m/\log m + \dein}$ upper bound
on the space complexity of an arbitrary DAG $\gra$, with $m$ edges and
maximum in-degree $\dein$.  This improves on the
$\BO{\dein
  n/\log n}$ upper bound by Loui~\cite{loui1979minimum}. The improvement can be significant when $\dein$ is
considerably larger than the average in-degree $m/n$.  For example, if
$m=\BT{n}$ and $\dein=\BT{\log_2n}=\BT{\log m}$, our upper bound
becomes $\BO{n/\log n}$, whereas the previous one is a trivial $\BO{n}$. Further, we provide an explicit construction of a pebbling schedule for $G$ using the stated number of pebbles. Our upper bound is existentially tight with respect to the maximum in-degree (trivially)
and also when $m=\BO{n}$ (matching a non trivial lower bound of
\cite{paul1976space}). 

While in many DAGs of interest the maximum degree is constant
and uniform, an interesting exception is provided by Deep Neural Networks (DNNs)
whose architecture is based on ``neurons'', with a number of inputs
dependent on the problem size. When trained based on the ``drop-out''
technique, several of the network weights are typically set to zero,
resulting in the removal of the corresponding inter-layer
connections. Thus, the final DNN DAG, used for inference, will exhibit
an unbalanced in-degree distribution.


With the developed techniques, we derive a $\BO{g\sqrt{n}+d}$ upper bound on the pebbling number for DAGs with genus $g$. In the case of planar DAGs ($g=1$), this improves on the $\BO{\sqrt{n}+d\log n}$ result by Lipton and Tarjan~\cite{liptarapp}.
We also derive a $2\sqrt{ml}-l +1+d$ upper bound on the pebbling number for DAGs of topological depth $l$, improving on the classic $l\left(d-1\right)+1$ upper bound for DAGs with
$m\leq \frac{d^2\mybrace{l-1}^2}{4l}$ edges.


\section{Preliminaries}
A \emph{Directed Acyclic Graph} (DAG) $\gra$ consists of a finite set
of \emph{vertices} $V$ and of a set of directed \emph{edges}
$E\subseteq V \times V$, forming no directed cycle. We say that the edge
$\left(u,v\right)\in E$ is directed from $u$ to $v$, \emph{exits} from
$u$, and \emph{enters} $v$.  Further, $u$ is said to be a
\emph{predecessor} of $v$ and $v$ a \emph{successor} of $u$. The \emph{in-degree} of a vertex $v$, written as $d(v)$, is the number of edges entering it. 
The maximum in-degree $d$ of $G$ is the maximum number of predecessors of any vertex in $V$. 

A vertex $w$ is an \emph{ancestor} (resp., \emph{descendent}) of $v$ if there exists a directed path from (resp. to) $w$ to (resp., from)  $v$. 
We denote as $\mathcal{G}\mybrace{n,m,d}$ the family of DAGs with $|V|\leq n,|E|\leq m$, and vertices with maximum in-degree at most $d$.
The sub-DAG of $G$ \emph{induced by} $V'\subseteq V$ is
$G'=\left(V',\left(V'\times V'\right)\cap E\right)$.
A \emph{topological ordering} $\phi=\mybrace{v_1,v_2,\ldots,v_n}$ of $G$ is a sequence of all its vertices where each vertex appears after all of its predecessors.

The pebble game~\cite{paterson1970comparative} is played on a DAG. An instance of the game, also called a \emph{pebbling schedule}, is a sequence of moves. In one move,
one can either (i) place a pebble on (or \emph{slid} it from one of the predecessors to) one vertex whose immediate predecessors all carry pebbles; or (ii) remove 
 a pebble from one vertex. A \emph{full schedule} starts with no pebbles on the DAG and, for any vertex $v$, includes a move that places a pebble on $v$. The maximum number of pebbles present on the DAG after a move is the \emph{space} of the schedule; the number of moves is the \emph{time}.  In this work, we develop algorithms that, given as input a DAG  $G\in\mathcal{G}\left(n,m,d\right)$, output a full pebbling schedule with guaranteed upper bounds on space and time.

\section{DAG Budget decomposition}\label{sec:decomposition}
Let $\phi= \left(v_1,\ldots,v_{n}\right)$ be a topological ordering of $G=\left(V,E\right)\in \mathcal{G}\mybrace{n,m, d}$. For $i\in\{1,\ldots,|\phi|\}$, the \emph{$i$-th boundary} $\boundary{G,\phi,i}$ is defined as the number of vertices in the prefix $\mybrace{v_1,\ldots,v_i}$ which have a successor in the suffix $\mybrace{v_{i+1},\ldots,v_n}$. The \emph{maximum boundary of }$\mybrace{G,\phi}$ is defined as $\boundary{G,\phi}\triangleq \max_{i\in\mycurl{1,\ldots,|\phi|}}\boundary{G,\phi,i}$. 

\sloppy Given $B\ge0$, a $B$-decomposition of $G$ is an ordered sequence of pairs $\mybrace{\mybrace{G_1, \phi_1},\ldots,\mybrace{G_\ell, \phi_\ell}}$, where $\phi=\mybrace{\phi_1\ldots\phi_\ell}$ is a topological ordering of $G$,  $V_i$ is the set of vertices  in $\phi_i$, $G_i=\mybrace{V_i,E_i}$ is the sub-DAG of $G$ induced by $V_i$, and $\sum_{i=1}^\ell \boundary{G_i, \phi_i} \leq B$. In Lemma~\ref{lem:decompositionspace}, we will show how a $B$-decomposition can be transformed into a pebbling schedule using at most $B+1+\mybrace{d-1}\mybrace{\ell-1}$ pebbles. A procedure to compute a $B$-decomposition, with $\ell$ bounded in terms of $B$, is given next.  Informally, $B$ can be viewed as a \emph{budget} that can be spent for the boundaries of the pieces of the decomposition.
The challenge is to minimize the number of pieces while complying with a given budget.\\

\noindent\textsc{Decompose}$\mybrace{G,\phi,B}$
\begin{itemize}
    \item If $\boundary{G,\phi}\leq B$, then return $\mybrace{G,\phi}$
    
    \item Else, let $i^*=\argmax_{i\in\mycurl{1,2,\ldots,n-1}} \boundary{G,\phi,i}$. Split $\phi$ into the prefix $\phi_p=\mybrace{v_1,\ldots,v_{i^*}}$ and the suffix $\phi_s=\mybrace{v_{i^*+1},\ldots,v_n}$. Let $G_p=(V_p,E_p)$ and $G_s=(V_s,E_s)$ the sub-DAGs of $G$ respectively induced by $V_p=\mycurl{v_1\ldots,v_{i^*}}$ and $V_s=\mycurl{v_{i^*+1},\ldots,v_{n}}$:
    \begin{itemize}
        \item return \textsc{Decompose}$\mybrace{G_p,\phi_p,B\frac{|E_p|}{|E_p|+|E_s|}}$, \textsc{Decompose}$\mybrace{G_s,\phi_s,B\frac{|E_s|}{|E_p|+|E_s|}}$.
    \end{itemize}
\end{itemize}
If \textsc{Decompose} is recursively invoked, then $|E_p|+|E_s| < |E|-B$, that is, more than $B$ edges are \emph{removed} from the subsequent steps. Further, for $B\geq n$, \textsc{Decompose} does not split the DAG and returns a $B$-decomposition with $\left(G,\phi\right)$ as the only element.

\begin{lemma}\label{lem:decomp}
Given $\gra{}\in \mathcal{G}\mybrace{n,m, d}$, any topological ordering of its vertices $\phi$, and a value $B\geq 0$, $\textsc{Decompose}\mybrace{G,\phi,B}$ constructs a $B$-decomposition of $G$  $\mybrace{\mybrace{G_1, \phi_1},\ldots,\mybrace{G_\ell,\phi_{\ell}}}$ with $\ell \leq \mymin{n,2^{\lfloor \frac{m}{B}\rfloor}}$. 
The construction takes $\BO{\mybrace{n+m}\mybrace{\frac{m}{B}+1}}$ time.
\end{lemma}
\begin{proof}
Consider the tree, $\mathcal{T}$, of the recursive calls executed by $\textsc{Decompose}$, whose root corresponds to the invocation   $\textsc{Decompose}\mybrace{G,\phi,B}$. The leaves of  $\mathcal{T}$ correspond
to the sub-DAGS of $G$, returned as the parts of the $B$-decomposition of $\mybrace{G,\phi}$. We first prove,  by induction on the number $j$ of levels of the tree, that $\mybrace{\mybrace{G_1, \phi_1},\ldots,\mybrace{G_\ell, \phi_\ell}}$ returned by $\textsc{Decompose}\mybrace{G,\phi,B}$ is indeed a $B$-decomposition of $G$. 
  
The base case, $j{=}1$, only occurs when  $\boundary{G,\phi}\leq B$ and $\textsc{Decompose}$ outputs $\mybrace{\mybrace{G,\phi}}$, which is trivially a 
$B$-decomposition of $G$.  For ${j>1}$, consider the subtrees of $\mathcal{T}$ rooted at the two children of the root, which correspond to the invocations  \textsc{Decompose}$\mybrace{G_p=\mybrace{V_p,E_p},\phi_p,B\frac{|E_p|}{|E_p|+|E_s|}}$ and \textsc{Decompose}$\mybrace{G_s=\mybrace{V_s,E_s},\phi_s,B\frac{|E_s|}{|E_p|+|E_s|}}$. Each of these subtrees has at most $j{-}1$ levels.
By definition, $\phi=\phi_p\phi_s$, $V_p$ and $V_s$ partition $V$, and $B\frac{|E_p|}{|E_p|+|E_s|}+B\frac{|E_s|}{|E_p|+|E_s|}= B$.  By the inductive hypothesis, the sequences $\mybrace{\mybrace{G_{p,1}, \phi_{p,1}},\ldots,\mybrace{G_{p,\ell_p}, \phi_{p,\ell_p}}}$ and  $\mybrace{\mybrace{G_{s,1}, \phi_{s,1}},\ldots,\mybrace{G_{s,\ell_s}, \phi_{p,\ell_s}}}$ returned by the calls respectively are a $\mybrace{B\frac{|E_p|}{|E_p|+|E_s|}}$-decomposition of $G_p$ and a $\mybrace{B\frac{|E_s|}{|E_p|+|E_s|}}$-decomposition of $G_s$. By construction, $\ell=\ell_p+\ell_s$. Further, by the inductive hypothesis,  $\phi_p = \phi_{p,1}\ldots,\phi_{p,\ell_p}$, $\phi_s = \phi_{s,1}\ldots,\phi_{s,\ell_s}$, with $$\sum_{i=1}^{\ell_p}\boundary{G_{p,i},\phi_{p,i}}+\sum_{i=1}^{\ell_s}\boundary{G_{s,i},\phi_{s,i}}\leq B\frac{|E_p|}{|E_p|+|E_p|}+B\frac{|E_s|}{|E_p|+|E_s|}=B.$$

We now prove the upper bound on $\ell$. By construction, each element of a $B$-decomposition includes at least one node; hence, $\ell\leq n$. Further, any decomposition with $n$ parts each with a single vertex is a $B$ decomposition for any $B\geq 0$, since, by definition, for a DAG $G'$ with a single vertex, $\boundary{G',\phi'}=0$.

Let a node at level $k$ of  $\mathcal{T}$ (the root being at level $0$) correspond to the invocation $\textsc{Decompose}\mybrace{H,\phi_H,B_H}$ where $H$ is a sub-DAG of $G$ with $m_H$ edges. We claim that 
\begin{equation}\label{eq:leaves}
m_H \leq B_H\mybrace{\frac{m}{B}-k}.    
\end{equation}
The argument is by induction on $k$. For $k=0$, Inequality~\ref{eq:leaves} clearly holds, as $H=G$, $m_H=m$, $\phi_H=\phi$, and $B=B_H$. If $k>0$, let $\mybrace{H',B'}$ and $\mybrace{H'',B''}$ denote the children of a node at level $k{-}1$, corresponding to the invocation $\textsc{Decompose}\mybrace{H,\phi_H,B_H}$. Let $m'$ and $m''$ be the number of edges of $H'$ and $H''$, respectively. By the definition of $\textsc{Decompose}$ and by the inductive hypothesis, we have:
\begin{equation*}
    m'+m''< m_H-B_H\leq B_H\mybrace{\frac{m}{B}-(k{-}1)} -B_H =
    B_H\mybrace{\frac{m}{B}-k},\end{equation*}
where the second passage follows from the inductive hypothesis. Then:

\begin{align*}
    \frac{m'}{B'} {=} \frac{m''}{B''}&= \frac{m'+m''}{B_H} < \frac{m}{B}-k.
\end{align*}
This establishes~\eqref{eq:leaves}. Considering that $0 \leq m/B-\lfloor m/B \rfloor <1$, Inequality~\eqref{eq:leaves} implies that nodes at level $\lfloor m/B\rfloor$ (if any) correspond to invocations of $\textsc{Decompose}$ for sub-DAGs for which the number of edges, which is an upper bound for the maximum boundary, is lower than or equal to the assigned budget. By construction, these calls will not trigger further invocations of $\textsc{Decompose}$. Thus, binary tree $\mathcal{T}$ has at most $\lfloor m/B \rfloor {+}1$ levels, hence at most
$\ell\leq 2^{\lfloor\frac{m}{B} \rfloor}$ leaves, each corresponding to a part
of the $B$-decomposition. 

Finally, let us analyze the running time of $\textsc{Decompose}$.
Given $G$ and $\phi=\mybrace{v_1\ldots v_n}$, the lowest index  $i\in\mycurl{1,\ldots,n-1}$ such that $\boundary{G,\phi,i}= \boundary{G,\phi,i^*}$
can be determined by the following procedure, which scans $\phi$ one vertex at a time, keeping track at each step of the set of vertices in the boundary, $W$, and of the index $i^*$ for which $W$ had the largest size thus far, $B^*$. 
Initially, $W=\emptyset$, $i^*=1$, and $B^*=0$. For $i {=} 1,\ldots,n{-}1$, at the $i$-th step,  if vertex $v_i$ has successors, then it is added to the boundary and a \emph{counter} associated with $v_i$ is initialized with the number of its successors. Then, the counter associated with each predecessor of $v_i$ is decremented by one. If a counter reaches zero, then the corresponding vertex is removed from the boundary. At this stage, if $|W| > B^*$, then $B^*$ is set to $|W|$ and $i^*$ to $i$. It is easy to see that the described procedure does a constant amount of work for each vertex and for each edge, thus running in time $\BO{n+m}$. Considering that (a) no more than $n$ nodes and $m$ edges overall are involved in the calls at each level of tree $\mathcal{T}$ and (b) that the number of levels is at most $\lfloor \frac{m}{B} \rfloor +1$, the stated bound on the running time of $\textsc{Decompose}$ remains established.
\end{proof}

\section{From budget decomposition to pebbling schedules}
In this and the following section, we show how to pebble any DAG $G \in \mathcal{G}\left(n,m,\dein\right)$ with $\mathcal{O}\left(m/\log m +\dein{}\right)$ pebbles. Our approach can be viewed as a depth-first execution of the DAG
according to a decomposition into sub-DAGs (obtained from a topological
sorting). The total space requirement of the sub-DAGs upper bounds the
space needed, at any given stage, to store the internal state of the
sub-Dag executions. The number of sub-DAGs, multiplied by the degree,
upper bounds the space needed to store data produced by a ``source''
sub-DAG, but not already consumed by the "destination" sub-DAG. Budget
decompositions aim to keep the tradeoff between these two terms
under control.

This and comparable results \cite{loui1979minimum,lengauer1979upper,hopcroft1977time} are of interest for $m/\log_2 m <n$, as $n$ pebbles trivially suffice to pebble any $n$-vertex DAG. The next lemma introduces a simple building block for pebbling schedules.
\begin{lemma}\label{lem:lem1} 
 Given a DAG $\gra{}\in\mathcal{G}\mybrace{n,m,d}$, let $\phi=\mybrace{v_1,v_2,\ldots,v_n}$ be a topological ordering of its vertices. There exists a pebbling schedule for $G$ with $T=2n$ moves, using $S \leq \boundary{G,\phi}+1$ pebbles.
\end{lemma}
\begin{proof} Consider the \emph{topological order pebbling strategy}, defined as the concatenation of sub-schedules $\mathcal{D}=\mathcal{D}_1^+\mathcal{D}_1^-
\mathcal{D}_2^+\mathcal{D}_2^- \ldots \mathcal{D}_n^+\mathcal{D}_n^-$, 
where $\mathcal{D}_i^+$ places a pebble on $v_i$ and $\mathcal{D}_i^-$
removes pebbles from any vertex in $\{v_1,\ldots,v_i\}$ with no successors
in $\{v_{i+1},\ldots,v_n\}$. Correctness: since $\phi$ is a topological ordering of $G$, all predecessors of  $v_i$ carry a pebble just before $\mathcal{D}_i^+$. $T=2n$ moves: each vertex is pebbled and unpebbled just once. Number of pebbles: by definition of $B$-decomposition, for each $i$, both just before $\mathcal{D}_i^+$ and just after $\mathcal{D}_i^-$, at most $\boundary{G,\phi}$ pebbles reside on vertices of $G$; only one additional pebble is used to pebble $v_i$.
\end{proof}
We now show how to build a schedule from a decomposition of the DAG.
\begin{lemma}\label{lem:decompositionspace}
Given $\gra{}\in\mathcal{G}\mybrace{n,m,d}$, let $\phi$ be a topological ordering of its vertices, and let $\mybrace{\mybrace{G_1=\mybrace{V_1,E_1},\phi_1},\ldots,\mybrace{G_\ell=\mybrace{V_\ell,E_\ell},\phi_\ell}}$ be a decomposition of $G$ with $\phi=\phi_1\phi_2\ldots\phi_\ell$, $V_i=\phi_i$, and $E_i=\mybrace{V_i\times V_i}\cap E$. Then, $G$ can be pebbled using
$S{\leq}\sum_{i=1}^\ell \boundary{G_i,\phi_i} +1+\mybrace{d-1}\mybrace{\ell-1}$ pebbles.
\end{lemma}
\begin{proof}
We inductively construct an \emph{emptying} pebbling schedule $\mathcal{C}$ for $G$, that is, one which leaves no pebbles on the graph, at the end.
    For $\ell{=}1$, $\mathcal{C}$  is the (emptying) schedule introduced in Lemma~\ref{lem:lem1}. For $\ell{=}1$, the bound on $S$ stated in this lemma equals that in Lemma~\ref{lem:lem1}.
    
For $\ell>1$, we define $\mathcal{C}$ inductively in terms of a schedule $\bar{\mathcal{C}}$ for the sub-DAG induced by 
$\cup_{i=1}^{\ell-1} V_i$, using $\bar{S}{\leq}\sum_{i=1}^{\ell-1}\boundary{G_i,\phi_i}{+}1{+}\mybrace{d{-} 1}\mybrace{\ell{-}2}$ pebbles
and of the schedule $\mathcal{D}$ for $(G_{\ell},\phi_{\ell})$ given by Lemma~\ref{lem:lem1}, using $S^{\mathcal{D}}\leq \boundary{G_{\ell},\phi_{\ell}}{+}1$ pebbles. For  $u\in V_{\ell}$, let $P(u)=\predc{u}\cap \cup_{i=1}^{\ell-1} V_i$.

If $P(u)=\emptyset$ for every $u \in V_{\ell}$, then  we define $\mathcal{C}{=}\bar{\mathcal{C}}\mathcal{D}$, which  pebbles all of $G$, is an emptying schedule, and uses $\max(\bar{S},S^{\mathcal{D}})\leq S$ pebbles.

If $P(u){\not=}\emptyset$ for at least one $u \in V_{\ell}$, $\mathcal{C}$ is obtained by a suitable interleaving of the moves of $\mathcal{D}$ with moves from suitably modified versions of $\bar{\mathcal{C}}$. Specifically, letting $\phi_{\ell}=(u_1,u_2, \ldots, u_s)$ and $\mathcal{D}=\mathcal{D}_1^+\mathcal{D}_1^-
\mathcal{D}_2^+\mathcal{D}_2^- \ldots \mathcal{D}_s^+\mathcal{D}_s^-$,
we construct $\mathcal{C}$ as
\[
\mathcal{C} = (\bar{\mathcal{C}}_1^+ \mathcal{D}_1^+\mathcal{D}_1^-\bar{\mathcal{C}}_1^-\bar{\mathcal{C}}_1')(\bar{\mathcal{C}}_2^+ \mathcal{D}_2^+\mathcal{D}_2^- \bar{\mathcal{C}}_2^-\bar{\mathcal{C}}_2') \ldots (\bar{\mathcal{C}}_s^+ \mathcal{D}_s^+\mathcal{D}_s^-\bar{\mathcal{C}}_s^-\bar{\mathcal{C}}_s'),
\]
where the sub-schedules $\bar{\mathcal{C}}_i^+$, $\bar{\mathcal{C}}_i^-$, and $\bar{\mathcal{C}}_i$ are defined next, for $i{=}1,2, \ldots, s$. 
Partition $\bar{\mathcal{C}}$ as $\bar{\mathcal{C}} {=} \bar{\mathcal{C}}_i \bar{\mathcal{C}}'_i$, where
$\bar{\mathcal{C}}_i$ is the shortest prefix of $\bar{\mathcal{C}}$ that pebbles all vertices of $P(u_i)$. $\bar{\mathcal{C}}_i$ is well defined, since $\bar{\mathcal{C}}$ does pebble all vertices in $\cup_{i=1}^{\ell-1} V_i \supseteq P(u_i)$. (If $P(u_i)$ is empty, then so is $\bar{\mathcal{C}}_i$.) 

Then, $\bar{\mathcal{C}}_i^+$ is obtained by eliminating from $\bar{\mathcal{C}}_i$, for each vertex $w \in P(u_i)$, the last move (if any) that removes a pebble from $w$. This guarantees that, at the end of $\bar{\mathcal{C}}_i^+$, all vertices in $P(u_i)$ carry pebbles.

Moreover, $\bar{\mathcal{C}}_i^-$ is defined as a sequence of moves, which remove the pebble from each vertex in $P(u_i)$ that does not carry a pebble at the end of $\bar{\mathcal{C}}_i$, when $\bar{\mathcal{C}}$ is executed. Thus, $\bar{\mathcal{C}}_i^+\bar{\mathcal{C}}_i^-$ yields, on $\cup_{i=1}^{\ell-1} V_i$, the same configuration of pebbles as $\bar{\mathcal{C}}_i$, whence $\bar{\mathcal{C}}_i^+ \bar{\mathcal{C}}_i^- \bar{\mathcal{C}}_i'$ is a legal schedule that pebbles all vertices of $\cup_{i=1}^{\ell-1} V_i$.

Clearly, $\mathcal{C}$ does pebble all of $G$ and is emptying.
To establish the upper bound on the pebble requirement, we observe that the number of pebbles on $V_{\ell}$ never exceeds $S^{\mathcal{D}}$ and the number of pebbles on $\cup_{i=1}^{\ell-1} V_i$ never exceeds $\bar{S}+(d{-}1)$, where the term $(d{-}1)$ accounts for the maximum number of vertices that may hold a pebble at some stage of $\bar{\mathcal{C}}_i^+$, while not holding a pebble in the corresponding stage of $\bar{\mathcal{C}}_i$.
When pebbling a $u_i$ for which $P(u_i) \not= \emptyset$, a pebble is slid from one
of its predecessors in $P(u_i)$.
Therefore, the number of pebbles on $V=\cup_{i=1}^{\ell} V_i$ is
at most $\bar{S}+(d{-}1){+}S^{\mathcal{D}}{-}1 = \sum_{i=1}^{\ell-1} \boundary{G_i,\phi_i} {+}1{+}\mybrace{d-1}\mybrace{\ell{-2}}+(d{-}1){+}\boundary{G_{\ell},\phi_{\ell}}{+}1-1$, which simplifies to the claimed bound for $S$.
\end{proof}

A DAG $G$ with $d=1$ can be straightforwardly pebbled with $S=1$  pebbles, in $T=2n$ moves, by decomposing any of its topological orderings into $1$-node segments.  Below, we develop bounds on space, time, and their tradeoffs, based on schedules built using a DAG decomposition. Results are often stated for $d \geq 2$
 and $m \geq 2$, in light of the preceding observation.

\begin{lemma}\label{lem:decompositiontime}
Given $\gra{}\in \mathcal{G}\mybrace{n,m,d}$, with $d\geq 2$, let $\phi$ be a topological ordering of its vertices. Let $\mybrace{\mybrace{G_1=\mybrace{V_1,E_1},\phi_1},\ldots,\mybrace{G_\ell=\mybrace{V_\ell,E_\ell},\phi_\ell}}$ be a decomposition of $G$ such that
$\phi=\phi_1\phi_2\ldots\phi_\ell$, $V_i=\phi_i$, and $E_i=\mybrace{V_i\times V_i}\cap E$. There exists a schedule $\mathcal{C}$ which pebbles $G$ in $T \leq 2\frac{d}{d-1}\mybrace{\frac{n}{\hat{\ell}}}^{\hat{\ell}}$ moves, with $\hat{\ell} = \mymin
{\ell,\lceil n/d \rceil,n/e}$, using
$
S \leq \sum_{i=1}^\ell \boundary{G_i,\phi_i} +1+\mybrace{d-1}\mybrace{\ell-1}
$ pebbles.
\end{lemma}
\begin{proof}
From the given decomposition of $G$, we obtain $\mybrace{\mybrace{H_1=\mybrace{W_1,F_1},\psi_1},\ldots,{\mybrace{H_{\ell^*}=\mybrace{W_{\ell^*},F_{\ell^*}},\psi_{\ell^*}}}}$, a decomposition where, for $i < \ell^*$, $|W_i| \geq d$, and
$
\sum_{i=1}^{\ell^*} \boundary{H_i,\psi_i} +1+\mybrace{d-1}\mybrace{\ell^*-1} \leq 
\sum_{i=1}^{\ell} \boundary{H_i,\psi_i} +1+\mybrace{d-1}\mybrace{\ell-1}.
$
Specifically, we iteratively merge the leftmost set with fewer than $d$ nodes with its immediate successor (if any) in the partition, until all sets, with the possible exception of the rightmost one, have at least $d$ nodes. Each merge increases the boundary size by at most $(d-1)$, an increment compensated by the decrement in the number of sets of the partition, in the overall bound.
Let now $\mathcal{C}$ be the schedule constructed in the proof of 
Lemma~\ref{lem:decompositionspace}, based on the new decomposition. Clearly,
the number of pebbles $S$ used by $\mathcal{C}$ satisfies the stated bound.

Next we show, by induction on $\ell^*$, that  $\mathcal{C}$ completes in at most $2\sum_{i=1}^{\ell^*}\prod_{j=i}^{\ell^*} |W_j|$ moves. For $\ell^*=1$, $W_1=V$ and $\mathcal{C}$ completes in $2|V|=2n$ moves, according to Lemma~\ref{lem:lem1}.
For $\ell^*>1$, recall that $\mathcal{C}$ pebbles (and, eventually, unpebbles) each vertex in $W_{\ell*}$ one single time. For any of these vertices with predecessors in $\cup_{i=1}^{\ell^*-1}W_i$, an execution of $\mathcal{C'}$ takes place which, by inductive hypothesis, completes within $2\sum_{i=1}^{\ell^*-1}\prod_{j=i}^{\ell^*-1}|W_j|$ moves. Hence, 
\begin{align*}
  T(\mathcal{C}) \leq |W_{\ell^*}|\mybrace{2\sum_{i=1}^{\ell^*-1}\prod_{j=i}^{\ell^*-1}|W_j|}+2|W_{\ell^*}| 
  &= 2\sum_{i=1}^{\ell^*}\prod_{j=i}^{\ell^*}|W_j| 
   = 2\prod_{j=1}^{\ell^*}|W_j| \sum_{i=1}^{\ell^*}\prod_{h=1}^{i-1}|W_h|^{-1},
\end{align*}
where we have established the inductive step and rewritten the result so that,
bounding $\prod_{j=1}^{\ell^*}|W_j|$ via the 
\emph{Arithmetic Mean/Geometric Mean} inequality
(given $\sum_{j=1}^{\ell^*}|W_j|=n$) as well as considering that, for $h<\ell^*$, $|W_h|^{-1} \leq d^{-1}$, yields:
\begin{align*}
  T(\mathcal{C}) \leq 
  2 \left(\frac{n}{\ell^*}\right)^{\ell^*} \sum_{i=1}^{\ell^*} d^{-(i-1)}    
  <2\frac{d}{d-1}\left(\frac{n}{\ell^*}\right)^{\ell^*}
  \leq 2\frac{d}{d-1}\left(\frac{n}{\hat{\ell}}\right)^{\hat{\ell}}.  \end{align*}
  Toward the last inequality, we recall that $\hat{\ell} = \mymin{\ell,\lceil n/d \rceil,n/e}$  and observe that (i) for $x \geq 0$, the function $f(x)=\left(n/x\right)^x$ increases with $x$ for $x \leq n/e$ and decreases elsewhere; (ii) $\ell^* \leq \ell$, and (iii) $\ell^* \leq \lceil n/d \rceil$.\\
  
\end{proof}
As shown in Lemma~\ref{lem:decomp}  for any DAG we can obtain in polynomial time a $B$-decomposition with at most $2^{\myfloor{\frac{m}{B}} }$ parts. Since $\mymax{e,\frac{n}{\ell}}^\ell$ is an non-decreasing function of $\ell$, by Lemma~\ref{lem:decompositiontime}:

\begin{lemma}\label{lem:boundbdecompo}
Any DAG $\gra{}\in \mathcal{G}\mybrace{n,m,d}$ can be pebbled in $T \leq 2\frac{d}{d-1}\mymax{e,\frac{n}{2^{\myfloor{\frac{m}{B}}}}}^{2^{\myfloor{\frac{m}{B}}}}$ moves using 
$
S \leq B+1+\mybrace{d-1}\mybrace{2^{\myfloor{\frac{m}{B}}}-1} 
$ pebbles,
where $B\geq 0$.
\end{lemma}

Given a DAG $\gra{}$ and a budget of pebbles $S$ to be used to pebble $G$, Lemma~\ref{lem:boundbdecompo}, naturally yields a strategy for constructing a schedule which (if possible) makes use of the allotted space by picking the maximum value $B$ such that $B+1+\mybrace{d-1}\mybrace{2^{\myfloor{\frac{m}{B}}}-1}\leq S$, if any, by constructing a corresponding $B$-budget decomposition of $G$, and then constructing a computation as discussed in Lemma~\ref{lem:boundbdecompo}. Such choice of the budget $B$ guarantees, in the worst case, the minimum running time achievable by a pebbling schedule constructed starting from a budget decomposition. 

As it is not in general trivial to compute such optimal value $B$ given $S$, in the following we consider some (generally sub-optimal) assignments of $B$ to obtain closed-form bounds:  

 \begin{theorem}[Upper bound on space for DAGs with in-degree at most $\log_2 m$]\label{thm:ubsingle}
There is an algorithm that, given any DAG $\gra{}\in \mathcal{G}\left(n,m,d\right)$ with $2 \leq d\leq \log_2 m$ and $m \geq 2^{12}$, constructs a schedule which pebbles $G$ in $T \leq 2\frac{d}{d-1}\mybrace{\frac{n(\log_2 m)^3}{m}}^{\frac{m}{(\log_2m)^3}}$ moves, using $S \leq \frac{m}{\log_2 m}+o\mybrace{\frac{m}{\log_2 m}}$ pebbles.
  \end{theorem}
  \begin{proof}
Let $m_0$ be such that, for $m \geq m_0$, $g(m)= \log_2 m -3 \log_2 \log_2 m \geq 1$ (e.g., $m_0=2^{12}$.) Clearly, any DAG with $m < m_0$ edges can be pebbled in time $T \leq 2n$ with $S_0 \leq m_0$ pebbles. For $m \geq m_0$, choose $B=\frac{m}{g(m)}$.
It is straightforward to check that (i) $B=\frac{m}{\log_2 m}\mybrace{1+o\mybrace{1}}$, (ii) $2^{\myfloor{\frac{m}{B}}} \leq 2^{g(m)}=\frac{m}{(\log_2 m)^3}$, and (iii) $\frac{n}{2^{\myfloor{\frac{m}{B}}}}\geq e$.
Using these relationships as well as $d \leq \log_2 m$ in the time and space bounds of Lemma~\ref{lem:boundbdecompo}, the time and space bound stated in the current lemma follow, after a few, simple manipulations.

Let $\phi$ be any topological ordering of the vertices in $G$. $\phi$ can be obtained in $\BO{n+m}$ time. By Lemma~\ref{lem:decomp}, it is possible to construct a $B$-budget decomposition of $G,\phi$ in at most $\BO{\mybrace{n-1+m}\mybrace{\frac{m}{B}}}= \BO{\mybrace{n-1+m}\mybrace{\log_2 m -  3\log_2 \log_2 m}}$ steps. Based on such decomposition, it is possible to construct a pebbling strategy $\mathcal{C}$ for $G$ as described in the proof of Lemma~\ref{lem:decompositionspace} in polynomial time.

\end{proof}

Theorem~\ref{thm:ubsingle} highlights how the number of pebbles required for DAGs with in-degree at most $\log_2 m$ is bounded from above by $m/\log_2 m$ plus lower order terms. 

\section{General upper bound on pebbling number}\label{sec:genuppbound}
Given a DAG $\gra{}$, let $G'=\mybrace{V\setminus W,E'}$ be the DAG obtained by removing all the edges with one or both endpoints in a selected subset $W \subseteq V$. Pebbling $G'$ may require considerably fewer pebbles than pebbling  $G$. This opens the door to a strategy to pebble $G$ by mimicking a schedule for $G'$ to pebble each of the ``\emph{challenging vertices}'' in $W$, one at a time in topological order, and maintaining them pebbled until the end of the schedule, so that they are available when pebbling their children.
The next lemma develops this approach and provides time and space bounds for $G$, in terms of those for $G'$.



\begin{lemma}[Challenging vertices]\label{lem:challenging}
Let $G$ be a DAG $\gra{}\in \mathcal{G}\left(n,m,d\right)$, let $W\subseteq V$ and let $G'=\left(V\setminus W,E'\right)$, where $E'=E \cap\left(\left(V\setminus W\right)\times\left(V\setminus W\right)\right)$. Given  a schedule $\mathcal{C}'$ which pebbles $G'$ in time $T'$ and space $S'$,  a schedule $\mathcal{C}$ can be constructed, which pebbles $G$ in time $T \leq (|W|+1)\mybrace{T'+n}$ and space $S \leq S'+|W|+d$.
\end{lemma}
\begin{proof}
Let $w_1,w_2,\ldots, w_{|W|}$  be the vertices of $W$ in some topological ordering, w.r.t. $G$. (That is, if $1 \leq i < j \leq |W|$, then $w_j$ \emph{is not} an ancestor of $w_i$, in $G$). We construct $\mathcal{C}$ as
\[
\mathcal{C} = (\mathcal{C}'_1\mathcal{P}_1\mathcal{D}_1)
(\mathcal{C}'_2\mathcal{P}_2\mathcal{D}_2) \ldots
(\mathcal{C}'_{|W|}\mathcal{P}_{|W|}\mathcal{D}_{|W|})
\mathcal{F},
\]
where, for $i=1,2, \ldots, |W|$,
$\mathcal{C}'_i$ is obtained from $\mathcal{C}'$ by skipping all moves targeting vertices that \emph{are not} proper ancestors of $w_i$ and all deletion moves targeting the parents of $w_i$; $\mathcal{P}_i$ pebbles $w_i$; and $\mathcal{D}_i$ removes any pebbles, except those on $w_1,\ldots, w_i$. Finally, $\mathcal{F}$ is obtained from $\mathcal{C}'$, by skipping all moves on $w_1,\ldots,w_{|W|}$. (It is easy to see that,  once it is pebbled, by $\mathcal{P}_i$, $w_i$ remains pebbled until the end of $\mathcal{C}$.) The objective of $\mathcal{F}$ is to pebble any $v \in V$ which is not an ancestor of any $w \in W$, since such $v$ has not been pebbled earlier.

It is easy to see that, for each $i$, $\mathcal{C}'_i\mathcal{P}_i$ makes at most $T'$ moves and $\mathcal{D}_i$ makes at most $n$ moves, while $\mathcal{F}$ makes at most $T'$ moves. Therefore, $\mathcal{C}$ makes $T \leq (|W|+1)\mybrace{T'+n}$ moves. At any stage, the vertices that hold a pebble under $\mathcal{C}$, but not under $\mathcal{C}'$, are in $W$ or are the (at most $d$) predecessor of some $w_i$. Hence,
$\mathcal{C}$ uses $S \leq S'+|W|+d$ pebbles.
\end{proof}
Next, we apply Lemma~\ref{lem:challenging} to derive bounds for a DAG $G$ of any degree $d$, by choosing $W$ as the set of vertices with in-degree greater than $\log_2 m$, so that DAG $G'$ falls within the scope of Theorem~\ref{thm:ubsingle}. Clearly, $|W|\leq \log_2 m$. The resulting space bound includes an additive term $d$, which cannot be avoided, as a vertex of in-degree $d$ can be pebbled only at a stage when all of its $d$ predecessors carry a pebble.\\


\begin{theorem}[General upper bound to pebbling number of DAGs]\label{thm:genuppbound}
There is an $\BO{\mybrace{n+m}\log m}$ time algorithm that, given any DAG $\gra{}\in \mathcal{G}\left(n,m,d\right)$, constructs a schedule to pebble $G$ in $T \leq \mathcal{O}\left(\log m\frac{d}{d-1}\mybrace{\frac{n\log_2^3 m}{m}}^{\frac{m}{\log_2^3 m}}\right)$ moves using $S \leq 2\frac{m}{\log_2 m}+o\mybrace{\frac{m}{\log_2 m}}+d$ pebbles.
\end{theorem}

\section{Comparison with previous results}\label{sec:comparison}
We break down the comparison with previous space results based on the parameters $n$, $m$, and $d$. If $n \leq m/\log_2 m$, then our bounds become trivial, as $n$ pebbles are sufficient to pebble any DAG. The same holds true for all other bounds in the literature~\cite{hopcroft1977time,loui1979minimum,lengauer1979upper}, which grow at least as $\dein{}n/\log_2 n >m/\log_2 m$.

If $m/\log_2 m < n \leq \dein n/\log_2 n$, our bound retains significance while the previous ones become trivial. This is the case, for example,  when $\dein_{avg} < \log_2 n \leq \dein$, where $\dein{}_{avg}=m/n$ denotes the \emph{average in-degree}.

Finally, when $\dein{} < \log_2 n$, both our bounds and the $3\dein n/\log_2 n + 4 $ upper bound by Loui in~\cite{loui1979minimum} (which exhibits the weakest dependence on $\dein{}$, among previous results) are non-trivial.  
To simplify the comparison with~\cite{loui1979minimum}, we modify the argument in the proof of Theorem~\ref{thm:ubsingle}, by choosing $B=2m/\log_2 m$, to obtain the following result,  for DAGs with $\dein\leq \left(\log_2 m\right)/3$. 

\begin{lemma}
\label{lem:uppbfinite}
There is an $\BO{\mybrace{n+m}\log m}$ time algorithm that, given any DAG $\gra{}\in \mathcal{G}\left(n,m,d\right)$, with $m>1$ and $d \leq \left(\log_2 m\right) /3$, constructs a schedule to pebble $G$ in
$T \leq 2\mybrace{\frac{n\log_2 m}{m}}^{\frac{m}{\log_2 m}}$ moves, using
$S \leq 2.8125\frac{m}{\log_2 m}$ pebbles.\end{lemma}

By Lemma~\ref{lem:uppbfinite}, we can conclude that for a given DAG $G\in\mathcal{G}\left(n,m,\dein\right)$, our algorithm provides a schedule such that the number of required pebbles is lower than the $3\dein n/\log_2 n + 4 $ upper bound by Loui in~\cite{loui1979minimum} by a factor $\frac{3\dein{}\log_2 m}{2.8125\dein{}_{avg}\log_2 n}$. The wider the gap between average and maximum in-degree, the higher the advantage of our method. But even when $\dein{}_{avg} \approx \dein{}$, our approach requires fewer pebbles.

\paragraph{Budget decompositions vs. layered partitions of a DAG} To compare the time of our schedules with the time of the schedules by Loui~\cite{loui1979minimum}, we highlight some relationships between our $B$-decompositions and  Loui's \emph{layered partitions} of a DAG with bounded cumulative \emph{internal overlap}.
The internal overlap $\omega\left(G\right)$ of $G=\left(V,E\right)$ is defined as the maximum, over all topological partitions $(V_1,V_2)$ of $V$, of the number of edges from $V_1$ to $V_2$.
Loui~\cite[Lemma 1, Theorem 3]{loui1979minimum} showed that for any topological partition (called \emph{layered partition} in~\cite{loui1979minimum})  there exists a schedule for $G$ using $\sum_{i=1}^k \omega\left(\left( W_i, \left(W_i\times W_i\right)\cap E\right) \right)+k\left(d+1\right)$ pebbles in 
\begin{equation*}
    T \leq 2n\left(1+\frac{dn}{\gamma S}\right)^k\frac{\gamma S}{dn},
\end{equation*}
moves, where $\gamma = \frac{2d+4}{\left(3d+2\right)\left(3d+3\right)}$. Loui's $\BO{dn/\log n}$ upper bound on pebbling follows from observing that for any DAG there exists a layered partition 
 $\left(W_1, \ldots, W_k\right)$ with $k\leq 2^{\lceil dn/r\rceil}$ and \emph{cumulative internal overlap} $r=\sum_{i=1}^k \omega\left(\left( W_i, \left(W_i\times W_i\right)\cap E\right)\right)$.

The construction of the pebbling schedule in~\cite[Theorem 3]{loui1979minimum} can be mimicked using a $B$-budget decomposition instead of a layered partition with cumulative internal overlap $B$, without affecting the bounds on the number of pebbles and on the number of moves. However, $B$-decompositions have two advantages over layered partitions. (i) $B$-layered partitions can be constructed efficiently (Lemma~\ref{lem:decomp}), whereas no explicit construction of layered partitions with bounded cumulative internal overlap is provided in~\cite{loui1979minimum}. (ii) For any DAG $G$, given a layered partition $\left(W_1,\ldots, W_k\right)$ with cumulative overlap $r$, a $r'$-budget decomposition $\left(G_1,\phi_1\right),\ldots,\left(G_{k},\phi_{k}\right)$ of $G$ with $r'\leq r$ can be constructed, where $G_i=\left(W_i,E\cap\left(W_i \times W_i\right)\right)$ and $\phi_i$ is \emph{any} topological ordering of $G_i$. By definition, $\boundary{G_i,\phi_i}\leq \omega\left(W_i\right)$.
Therefore, the layered partition corresponding to a
given $B$-budget decomposition may have cumulative internal overlap higher than $B$. 
A better tradeoff between $B$ and $k$ is of interest as these quantities affect the number of moves of the schedule
constructed using those as the basis (the lower the number of components the lower the number of moves). 

In conclusion, if $S$ pebbles are sufficient to construct a schedule of $T$ moves, based on a layered partition, then $S$ pebbles are also sufficient to efficiently construct a schedule of at most $T$ moves, based on a budget decomposition.


\section{Other applications of the challenging vertices technique}

\subsection{DAGs with bounded genus}
A DAG is \emph{planar} if it can be drawn on the plane in such a way that its edges intersect only at their endpoints. If $G = (V,E)$ is planar, then $|E|<3|V|$.
Lipton and Tarjan~\cite{liptarapp} showed that 
$\mathcal{O}\left(\sqrt{n} + d\log n\right)$ pebble are sufficient to pebble any planar DAG $G\in \mathcal{G}\left(n,m,d\right)$. This result is based their ``\emph{Planar Separation Theorem}''\cite{Lipton1977AST}: if $G = (V,E)\in \mathcal{G}\left(n,m,d\right)$ is planar, then $V$ can be partitioned into three subsets $V_L,V_R$, and a \emph{vertex separator} $V_S$ such that:
(1) There are no edges $\left(u,v\right)$ such that $u\in V_L$ and $v\in V_R$ or vice-versa;
    (2) $\mymax{|V_L|,|V_R|}\leq \frac{2}{3}n$;
    (3) $|V_S|\leq 2\sqrt{2}\sqrt{n}$. Using the challenging vertices techniques we obtain the following result,
    which is asymptotically tighter than the one in~\cite{liptarapp} if $\sqrt{n}/\log n = o\mybrace{d}$ and is equivalent otherwise.


    

\begin{theorem}\label{thm:planar}
    Any planar $G\in \mathcal{G}\mybrace{n,m,d}$ can be pebbled with
    $s\left(n,d\right)=6\mybrace{\sqrt{2}+\sqrt{3}}\mybrace{1+\sqrt{\frac{2}{3}}}\sqrt{n} + d$ pebbles.
\end{theorem}
\begin{proof}
    (By induction on $n$.)  For $n=1$ the schedule that simply pebbles the only vertex satisfies the statement. 
    For $n>1$, partition $V$ into $V_L,V_R$, and $V_S$ according to the planar separator, which guarantees the properties recalled above. Also, let $Z=\{v\in V\ s.t.\ d(v)\geq \sqrt{3n}\}$ which, since $|E| < 3|V|$, satisfies $|Z| < \sqrt{3n}$. The objective is to invoke Lemma~\ref{lem:challenging}, with $W=V_S \cup Z$. To this end, we consider $G'=\left(V\setminus W,E \cap\left(\left(V\setminus W\right)\times\left(V\setminus W\right)\right)\right)$. By construction, $G'$ has two separate components $G'_L=\left(V_L, E_L\right)$ and $G'_R=\left(V_R, E_R\right)$ which are both planar and in $\mathcal{G}\left(2n/3,m', \sqrt{3n}\right)$. As the components can be pebbled one at a time, by the inductive hypothesis applied to $G'_L$ and $G'_R$, we have that $G'$ can be pebbled using at most $s\mybrace{\frac{2}{3}n,\sqrt{3n}}=6\mybrace{\sqrt{2}+\sqrt{3}}\mybrace{1+\sqrt{\frac{2}{3}}}\sqrt{\frac{2}{3}n} +\sqrt{3n}$ pebbles. 
    By Lemma~\ref{lem:challenging}, $G$ can be pebbled with a number of pebbles:
    \begin{align*}
        S &\leq s\mybrace{\frac{2}{3}n,\sqrt{3n}}+|W|+d \\&=6\mybrace{\sqrt{2}+\sqrt{3}}\mybrace{1+\sqrt{\frac{2}{3}}}\sqrt{\frac{2}{3}n} +\sqrt{3n} +2\sqrt{2}\sqrt{n} + \sqrt{3n} +d= s\left(n,d\right).
    \end{align*}
\end{proof}
Gilbert \emph{et al.}~\cite{gilbert82} proved that, for any $n$-vertex graph $G=\left(V,E\right)$ embedded on an orientable surface of genus $g$, $V$ can be partitioned into subsets $V_L,V_R$ and $V_S$ such that:
(1) There are no edges $\left(u,v\right)$ such that $u\in V_L$ and $v\in V_R$ or vice-versa;
    (2) $\mymax{|V_L|,|V_R|}\leq \frac{2}{3}n$;
    (3) $|V_S|\leq 2g\sqrt{n}+ \sqrt{8n}$; (4) DAGs $G_L=(V_L,(V_L\times V_L)\cap E)$ and $G_L=(V_R,(V_R\times V_R)\cap E)$ are planar. Using this result, Theorem~\ref{thm:planar} can be generalized to yield an $\mathcal{O}\left(g\sqrt{n}+d \right)$ upper bound on the pebbling number of any DAG with genus $g$. 

Similarly, using the separator of size $\sqrt{\Tilde{g}n}$
established by Eppstein ~\cite{eppstein2002dynamicgeneratorstopologicallyembedded} for graphs with bounded Euler genus $\Tilde{g}$, we can obtain a $\mathcal{O}\left(\sqrt{\Tilde{g}n}+d \right)$ upper bound on the pebbling number of such graphs.

\subsection{Pebbling number upper bound based on DAG topological depth}
The \emph{topological depth} $l$ of a DAG is defined as the maximum number of edges in a directed path from an input to an output vertex. As well-known~\cite{savage97models},  $l\mybrace{d-1}+1$ pebbles are sufficient to pebble any DAG with topological depth $l$ and maximum in-degree $d$. Using the approach of challenging vertices,  a result tighter than the classical one can be obtained for DAGs with $m\leq \frac{d^2\mybrace{l-1}^2}{4l}$:
\begin{theorem}\label{thm:depth}
    Let $G\in\mathcal{G}\mybrace{n,m,d}$ with topological depth $l$. There exists a pebbling schedule for $G$ using at most $2\sqrt{ml}-l+1+d$ pebbles.
\end{theorem}
\begin{proof}
Let $W=\{v\in V\ s.t.\ d(v)\geq \sqrt{m/l}\}$. Clearly $|W|\leq m/\sqrt{m/l}\leq \sqrt{ml}$
Let $G'$ be defined as in Lemma~\ref{lem:challenging}. By construction, $G'$ has topological depth at most $l$ and  maximum in-degree $d'<\sqrt{m/l}$. Hence, the classic bound for $G'$ is  $l\mybrace{\sqrt{m/l}-1}+1$ pebbles. The statement follows from Lemma~\ref{lem:challenging}.
\end{proof}

\section{Conclusion}
In this paper, we described an efficient construction of a pebbling schedule for any DAG $G\left(V, E\right)\in\mathcal{G}\left(n,m,d\right)$ using at most $\min\{n, \BO{m/\log m + d}\}$ pebbles.

In our result, the maximum in-degree $d$ of the DAG appears as an additive term to the main $m/\log m$ component of the upper bound rather than a multiplicative term as in results previously presented in literature~\cite{loui1979minimum,hopcroft1977time}. 
As $n\leq m\leq dn$, if for a given DAG we have $m= \Theta(dn)$ our bound corresponds to the  $\BO{d \frac{n}{\log n}}$ upper bound in ~\cite{loui1979minimum,hopcroft1977time}. Note however that said bounds quickly go to $\BO{n}$ if just one vertex has an in-degree greater or equal to $\log n$, therefore losing significance. Our bound may still retain significance for some DAGs for which $m = o\left(dn\right)$, even if a limited number of vertices has in-degree greater than $\log n$.

The \emph{challenging vertices} technique also yields improved pebbling upper bounds for DAGs with bounded genus and DAGs with bounded topological depth.
 \bibliographystyle{elsarticle-num} 
 \bibliography{bibliography}

\end{document}